\documentclass[conference]{IEEEtran}

\usepackage{graphicx}
\usepackage{tikz}
    \usetikzlibrary{positioning}
    \usetikzlibrary{calc}
    \usetikzlibrary{shapes}
    \usetikzlibrary{patterns}
\usepackage[colorlinks=true,bookmarks=false,citecolor=blue,urlcolor=blue]{hyperref}
\usepackage{color, colortbl}
\usepackage{pgfplots}
    \pgfplotsset{compat=1.15}
    \usepgfplotslibrary{fillbetween}
    \usepackage{pgfplotstable}
\usepackage{dsfont}
\usepackage{amsmath,amssymb,amsthm}
    \theoremstyle{definition}
	\newtheorem{lemma}{Lemma}
	\newtheorem{theorem}{Theorem}
	
	\newtheorem{example}{Example}
    \newtheorem{remark}{Remark}
\usepackage{standalone}
\usepackage[linesnumbered,ruled]{algorithm2e}
\usepackage{cite}
\usepackage{physics}
\definecolor{myDarkGreen}{rgb}{0.00000,0.48824,0.00000}%

\def\lf{\left\lfloor}
\def\rf{\right\rfloor}
\def\lc{\left\lceil}
\def\rc{\right\rceil}

\def\define{\stackrel{\Delta}{=}}

\def\n{N}

\newcommand{\ux}{\underline{x}}

\newcommand{\uv}{\underline{v}}

\tikzstyle{point-visible} = [debugpoint,inner sep=0pt, minimum size = 2,fill=red] 
\tikzstyle{point-invisible} = [coordinate]
\tikzstyle{point} = [point-invisible] 
\colorlet{amp3color}{green!70!black}
\colorlet{amp1color}{blue!70!black}
\colorlet{amp5color}{red!70!black}
\colorlet{amp7color}{orange!90!black}
\colorlet{ESScolor}{green!70!black}
\colorlet{SMcolor}{red!70!black}
\colorlet{CCDMcolor}{red}
\colorlet{MPDMcolor}{blue}
\colorlet{SSHcolor}{black}

\title{Log-CCDM: Distribution Matching via Multiplication-free Arithmetic Coding}

\author{\IEEEauthorblockN{Yunus Can G\"{u}ltekin, Frans M. J. Willems, Alex Alvarado}
\IEEEauthorblockA{Information and Communication Theory Lab, Eindhoven University of Technology, The Netherlands\\
Email: y.c.g.gultekin@tue.nl}
}

\def\step{s}
\def\stepb{S}
\def\lgp{\text{Lg}^+}
\def\lgm{\text{Lg}^-}

\begin{document}

\maketitle

\begin{abstract}
    Recent years have seen renewed attention to arithmetic coding (AC).
    This is thanks to the use of AC for distribution matching (DM) to control the channel input distribution in probabilistic amplitude shaping.
    There are two main problems inherent to AC: (1) its required arithmetic precision grows linearly with the input length, and (2) high-precision multiplications and divisions are required.
    Here, we introduce a multiplication-free AC-based DM technique via three lookup tables (LUTs) which solves both problems above.
    These LUTs are used to approximate the high-precision multiplications and divisions by additions and subtractions.
    The required precision of our approach is shown to grow logarithmically with the input length.
    We prove that this approximate technique maintains the invertibility of DM.
    At an input length of 1024 symbols, the proposed technique achieves negligible rate loss ($<0.01$ bit/sym) against the full-precision DM, while requiring less than 4 kilobytes of storage.
\end{abstract}

%%%%%%%%%
\section{Introduction}\label{sec:intro}
%%%%%%%%%
Distribution matching (DM) converts a binary string into a symbol sequence with the desired distribution. 
DM is an essential block in probabilistic amplitude shaping (PAS) that controls the channel input distribution~\cite{BochererSS2015_ProbAmpShap}.
Thanks to this control, PAS achieves the capacity of the additive white Gaussian noise (AWGN) channel~\cite{Bocherer2014_ProbSigShapForBMD,Gultekin2020_RandomSignCoding}. 
PAS has also become popular in wireless~\cite{Schulte2019_SMDM,GultekinHKW2019_ESSforShortWlessComm} and fiber optical communications \cite{Buchali2016_RateAdaptReachIncrease,FehenbergerABH2016_OnPSofQAMforNLFC,idler2017_fieldtrialPS,Amari2019_IntroducingESSoptics}.

The DM technique used when PAS was introduced was constant composition distribution matching (CCDM)~\cite{SchulteB2016_CCDM}.
In CCDM, binary strings---which are the data sequences in this context---are converted into amplitude sequences with a fixed composition.
By selecting this composition such that the resulting amplitude distribution resembles a one-sided Gaussian distribution and by selecting the signs uniformly, a Gaussian-like channel input is obtained, which increases the achievable rates for the AWGN channel~\cite{BochererSS2015_ProbAmpShap} and optical channels~\cite{FehenbergerABH2016_OnPSofQAMforNLFC}.

When CCDM was first introduced~\cite{SchulteB2016_CCDM}, it was implemented using arithmetic coding (AC)~\cite{Ramabadran1990_moutofn}.
AC is a source coding algorithm that represents nonuniform source sequences by subintervals of the unit interval~\cite[Sec. 6.2]{MacKay_IT_Inf_Learning}.
Since DM is the dual operation to compression~\cite[p.222]{Fischer2002_PrecodingShaping}, AC can be used to realize DM in an ``inverted encoder-decoder pair'' setup~\cite{Martin1983ACforconstrainedchan} as illustrated in Fig.~\ref{fig:OverlapEffect}: matching is realized via AC decoding (decompression), dematching is realized via AC encoding (compression).

There are two main problems inherent to AC.
First, the required arithmetic precision grows linearly with the input length, which makes AC very challenging to realize for inputs longer than a thousand symbols.
Long blocklengths are necessary to decrease the rate loss of CCDM~\cite[Fig. 2]{SchulteB2016_CCDM}.
Second, multiplications and divisions are necessary to realize the algorithm.
There is a large amount of research dedicated to solving these problems in the context of source coding, see e.g.,~\cite{Tjalkens1987_DataCompressionthesis} and references therein.
For DM, the first problem is solved via a finite-precision implementation in~\cite{PikusXK2019_FPACDM}.
To the best of our knowledge, there is no study solving the second problem for DM.

\begin{figure}[t]
    \centering
\resizebox{0.64\columnwidth}{!}{\includegraphics{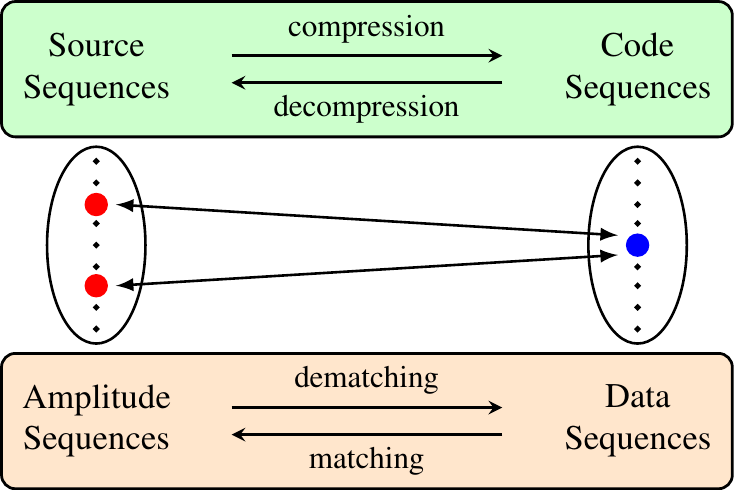}}
    \caption{Duality of source coding (top) and DM (bottom), and the effect of having overlapping subintervals in arithmetic coding (middle).}
    \label{fig:OverlapEffect}
\end{figure}

In this paper, we introduce ``Log-CCDM'', an approximate logarithmic (log) domain AC-based algorithm that solves both problems above.
Log-CCDM works based on three lookup tables (LUTs). 
High-precision multiplications and divisions required in AC are approximated by low-precision additions and subtractions in the log domain.
The required arithmetic precision of Log-CCDM grows logarithmically with the input length.
We prove that the invertibility of DM is maintained.
We demonstrate that the performance of this approximate algorithm (in terms of rate) depends on the sizes of the LUTs.
Requiring less than a few kilobytes (kBs), the rate loss of Log-CCDM against full-precision CCDM (FP-CCDM) is negligible, i.e., $<0.01$~bit/sym. 
This performance is achieved using 20-bit arithmetic operations which is comparable to that of the finite-precision algorithm of~\cite{PikusXK2019_FPACDM} that requires multiplications and divisions.

%%%%%%%%%%%%%%%%%%%%%%%%%%%%%%%%%%%%%%%%%%%%%%%%%%%%%%
\section{Distribution Matching via Arithmetic Coding}\label{sec:ACDM}
%%%%%%%%%%%%%%%%%%%%%%%%%%%%%%%%%%%
Consider a composition $C = [n_0, n_1,\dotsc, n_{A-1}]$ of symbols $m \in \mathcal{A} = \{0, 1,\dotsc, A-1\}$, resp., where $\sum_{a\in\mathcal{A}} n_a = N$ and $\mathcal{A}$ the symbol alphabet.
The corresponding set $\mathcal{C}_{\text{cc}}$ of CC symbol sequences $\ux = (x_1, x_2,\dotsc, x_N)$ consists of all sequences that have $\sum_{i=1}^{N} \mathds{1}[x_i=a] = n_a$ for $a\in\mathcal{A}$. 
Here $\mathds{1}[\cdot]$ is the indicator function which is $1$ when its argument is true and $0$ otherwise.
CCDM is a shaping technique that maps $k$-bit strings (indices) $\uv = (v_1, v_2,\dotsc, v_k)$ to CC symbol sequences $\ux \in \mathcal{C}_{\text{cc}}$~\cite{SchulteB2016_CCDM}.
This is called {\it matching}, while the reverse operation is called {\it dematching}.
The matching rate is then defined as $k/N$ bit/sym.
CCDM can be realized via AC.

In AC, the binary string $\uv$ is represented by a number $d(\uv)\in [0, 1)$ which is defined as $d(\uv) \define \sum_{i=1}^k v_i2^{-i}$.
Further, each symbol sequence $\ux$ corresponds to a subinterval $I(\ux)$ of the interval $[0, 1)$.
These subintervals partition the interval $[0, 1)$.
During arithmetic encoding, the interval $I(\ux)$ is found based on $\ux$, then a number $d(\uv) \in I(\ux)$ is determined. The outcome of this process is $\uv$.
During arithmetic decoding, the input $\uv$ is mapped to sequence $\ux$ if $d(\uv) \in I(\ux)$.

The interval $I(\ux)$ can be found based on the probability model $p_n(a) \define \Pr\{X_n=a|X_1,X_2,\dotsc,X_{n-1}\}$ for $n = 1, 2,\dotsc, N$ and $a\in\mathcal{A}$.
The {\it base} $b_N$ and {\it width} $w_N$ of the interval $I(\ux) = [b_N, b_N+w_N)$ are computed via the recursions
\begin{IEEEeqnarray}{rCl}
b_n &=& b_{n-1} + w_{n-1}\sum_{a<x_n} p_n(a),  \label{eq:b_update}\\
w_{n} &=& p_n(x_n)w_{n-1}, \label{eq:w_update}
\end{IEEEeqnarray}
for $n = 1, 2,\dotsc, N$ where $b_0=0$ and $w_0=1$.
From \eqref{eq:w_update}, we see that the width $w_N = \prod_{n=1}^{N} p_n(x_n)$ which is equal to $p(\ux)$ via the chain rule.
Thus, recursions \eqref{eq:b_update}--\eqref{eq:w_update} result in intervals $I(\ux)$ of length $p(\ux)$.
Typically, $d(\uv)\in I(\ux)$ is selected as the number that leads to the $\uv$ with the shortest representation~\cite[Sec. 5.2.3.1]{Sayood2002_LosslessCompression}.
The minimum length of $\uv$ is $\lc -\log I(\ux) \rc$ bits, and AC generates variable-length $\uv$ in general.
Here, $\log$ denotes the binary logarithm.

In the case of CCDM, the order of coding operations is inverted.
At the transmitter, arithmetic decoding is used to map $\uv$ to a $\ux\in\mathcal{C}_{\text{cc}}$, while arithmetic encoding is realized for the reverse mapping at the receiver.
For CCDM, the probability model $p_n(a)$ can be obtained considering the composition $C$ of the sequence $\ux$ and the composition of already-processed symbols $(x_1, x_{2},\dotsc, x_{n-1})$.
It is given by 
\begin{equation}
    p_n(a) = \frac{n_a - \sum_{i=1}^{n-1} \mathds{1}[x_i=a]}{N-n+1}. \label{eq:ccdmprobmodel}
\end{equation}
Via \eqref{eq:ccdmprobmodel}, the symbol probabilities are initialized as $p_1(a) = n_a/N$, and recursively updated by decreasing $n_a$ whenever a symbol $x=a$ is processed.
From \eqref{eq:ccdmprobmodel}, we see that $p(\ux) = \prod_{n=1}^N p_n(x_n) = (\prod_{a=0}^{A-1} n_a!)/N!=1/|\mathcal{C}_{\text{cc}}|$ for all $\ux\in\mathcal{C}_{\text{cc}}$.
A one-to-one mapping from $\uv$ to $\ux\in\mathcal{C}_{\text{cc}}$ is established if each $I(\ux)$ contains at most one $d(\uv)$.
Since the intervals are of equal length $1/|\mathcal{C}_{\text{cc}}|$, the maximum $k$ which guarantees this is $k = \lf \log |\mathcal{C}_{\text{cc}}| \rf$, and AC becomes fixed-length coding.

%%%%%%%%%%%%%%%%%%%%%%%%%%%%%%
\begin{example}{\bf (Binary-output CCDM)}
Consider the composition $C = [3, 2]$ for $\mathcal{A} = \{0, 1\}$, i.e., we want to generate binary sequences of length $N=5$ containing $2$ ones.
There are $10$ such sequences and $k=3$.
As an example, let us find the CC sequence $\ux$ that is mapped to $\uv = (1,1,0)$ ($d(\uv)=0.75$) which is written with blue in Fig.~\ref{fig:MatchExample}.
The first symbol $x_1$ partitions the interval $[0, 1)$ in fractions $p_1(0)= n_0/N = 3/5$ and $p_1(1) = n_1/N = 2/5$.
Since $d(\uv)=0.75 \in [6/10, 1)$, this interval is selected and $x_1=1$.
Then the second symbol $x_2$ partitions the interval $[6/10, 1)$ in fractions $p_2(0)= n_0/(N-1) = 3/4$ and $p_2(1) = (n_1-1)/(N-1) = 1/4$.
Since $d(\uv)=0.75 \in [6/10, 9/10)$, this interval is selected and $x_2=0$.
This procedure is repeated until $x_5$ and the final sequence $\ux=(1,0,0,1,0)$ is obtained.
\end{example}
%%%%%%%%%%%%%%%
\begin{figure}[t]
    \centering
\resizebox{\columnwidth}{!}{\includegraphics{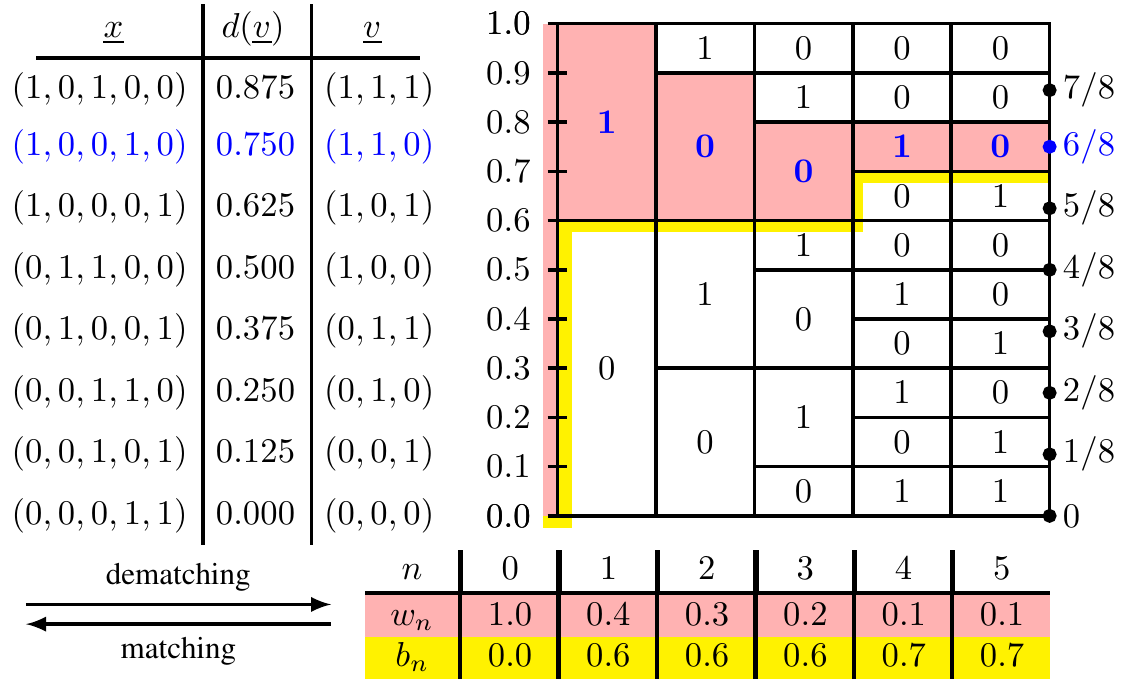}}
    \caption{CCDM for $C = [3, 2]$ and $A=2$ via AC.
    \textcolor{blue}{Blue} sequences $\ux = (1,0,0,1,0)$ and $\uv = (1,1,0)$ (with $d(\uv)=0.75)$ are mapped to each other.
    Base $b_n$ follows the yellow path, while the width $w_n$ is the height of the red region, see \eqref{eq:b_update}--\eqref{eq:w_update}.
    Sequences $(0,1,0,1,0)$ and $(1,1,0,0,0)$ are not utilized since the corresponding intervals $[0.4, 0.5)$ and $[0.9, 1.0)$ do not contain a number $d(\uv)$ for a 3-bit $\uv$.}
    \label{fig:MatchExample}
\end{figure}
%%%%%%%%%%%%%%%%%%%%%%%%%%%%%%
%%%%%%%%%%%%%%%%%%%%%%%%%%%%%%
\begin{algorithm}[t]
\KwIn{Index $\uv$, composition $C = [n_0, n_1]$ }
\KwOut{Sequence $\ux \in \{0, 1\}^{\n}$}
{\it Initialize:} $w_0 \gets 1$, $I_0\gets d(\uv)$, $N\gets n_0+n_1$ \\
\For{$n = 0, 1,\dotsc, \n-1$}{
\uIf{$I_n\geq w_nn_0/(N-n)$}{
$x_{n+1} \gets 1$ \\
{\color{red}$I_{n+1} \gets I_{n} - w_nn_0/(N-n)$}\\
{\color{blue}$w_{n+1} \gets w_{n}n_1/(N-n)$}\\
$n_1 \gets n_1-1$ \\
}
\Else{
$x_{n+1} \gets 0$ \\
$I_{n+1} \gets I_{n}$ \\
{\color{blue}$w_{n+1} \gets w_{n}n_0/(N-n)$}\\
$n_0 \gets n_0-1$ 
}
}
\KwRet{$\ux = (x_1, x_2,\dotsc, x_{\n})$}
\caption{FP-CCDM (Matching)}
\label{alg:ccdmatching}
\end{algorithm}
%%%%%%%%%%%%%%%%%%%%%%%%%%%%%%

During matching, the interval $I(\ux)$ that includes $d(\uv)$ can also be found without storing the base $b_n$, but instead, by successively subtracting it from $d(\uv)$.
The corresponding dematching algorithm can also be realized straightforwardly.
A pseudo-code for this FP-CCDM is given in Algorithm~\ref{alg:ccdmatching} for $\mathcal{A} = \{0, 1\}$.
The recursive subtraction of base from the input index is realized in line 5 (highlighted in {\color{red}red}).
This algorithm can also be realized in the log domain, which is the main idea behind Log-CCDM.
Then the width would be updated by an addition and a subtraction, instead of the multiplication and the division in lines 6 and 11 (highlighted in {\color{blue}blue}).

%%%%%%%%%%%%%%%%%%%%%%%%%%%%%%%%%%%%%%%%%%%%%%%%%%%%%%
\section{Log-CCDM with LUTs}\label{sec:LogCCDM}
%%%%%%%%%%%%%%%%%%%%%%%%%%%%%%%%%%%%%%%%%%%%%%%%%%%%%%
We will explain Log-CCDM for the binary case for simplicity.
Binary DM can be used to approximately shape the channel inputs~\cite{FabregasM2010_BICMwShaping,Pikus2017_BLDM,SteinerSB2018_PDM,PikusX2019_ACbasedMCBLDM,Gultekin2019_PESS}.
Extension to nonbinary alphabets is possible for all the techniques discussed in this section.

%%%%%%%%%%%%%%%%%%%%%%%%%%%%%%%%%%%%
\subsection{Log-CCDM}\label{ssec:logccdm}
%%%%%%%%%%%%%%%%%%%%%%%%%%%%%%%%%%%%
We define an exponential function for positive integer $s$
\begin{equation}
    F(\step) \define 
    \begin{cases}
    \lc M 2^{-\step/\stepb} \rc  &\mbox{ if }\hspace{0.05cm} 1 \leq s \leq S, \\
    \lc M 2^{-r/\stepb} \rc 2^{-d} \hfill &\mbox{ if }\hspace{0.05cm} s=r+dS>S, \\
    \end{cases}
    \label{eq:expfunc}
\end{equation}
where integers $r>0$, $d>0$.
In \eqref{eq:expfunc}, both $\stepb$ and $M$ are positive integers.
We see from \eqref{eq:expfunc} that $F(\step)$ can be computed for any positive integer $\step$ by storing $F(\step)$ only for $\step = 1, 2,\dotsc, \stepb$.
This requires a LUT with $\stepb$ entries. 
For $s>\stepb$, $F(s)$ can be computed only with shifts in base-2 thanks to the $2^{-d}$ factor in \eqref{eq:expfunc}.
We assume that $M$ is an integer power of two, and each entry of this LUT is stored with $\log M$ bits.
Note that with $\log M$ bits, only the nonnegative numbers below $M$ can be stored in an exact manner.
This, i.e., $F(s)<M$ for $1 \leq s \leq S$, is ensured when $S<M$ from \eqref{eq:expfunc}.
An example of $F(\step)$ is shown in Fig.~\ref{fig:exptable}.

Consider Algorithm~\ref{alg:ccdmatching} realized in the binary domain.
Here, $0\leq I_n<1$ and $0\leq w_n<1$  for $n = 0, 1,\dotsc, N$ and hence, their integer parts are always $0$.
Thus, we neglect their integer part and focus on their fractional part.
We approximate the fractional part of $w_n$ by $F(s_n)$ where $s_0\define 1$.
At step $n$ of Algorithm~\ref{alg:ccdmatching}, there are two possible choices for the width $w_{n+1}$ corresponding to symbols $1$ and $0$, resp., see lines 6 and 11.
We approximate these choices by first defining
\begin{equation}
    s_{n+1} \define
    \begin{cases}
        s_n-S\log(n_1/(n_0+n_1)) \quad \mbox{if} \quad x_{n+1}=1, \\
        s_n-S\log(n_0/(n_0+n_1)) \quad \mbox{if} \quad x_{n+1}=0. \\
    \end{cases}
    \label{eq:additiveupdateexact}
\end{equation}
Then we observe from \eqref{eq:expfunc} that $F(r+dS)=2^{-d}F(r)$ which implies $F(r-S\log K)= KF(r)$ for integer $S\log K$.
Then
\begin{equation}
F(s_{n+1}) =
\begin{cases}
F\left(s_n - \stepb\log\left(\frac{n_1}{n_0+n_1}\right)\right) \approx F(s_n)\frac{n_1}{n_0+n_1}, \\
F\left(s_n - \stepb\log\left(\frac{n_0}{n_0+n_1}\right)\right) \approx F(s_n)\frac{n_0}{n_0+n_1}, \\
\end{cases}
    \label{eq:subdivideF}
\end{equation}
for $x_{n+1}= 1, 0$, resp., for $n = 1, 2,\dotsc, N$.
In the context of Algorithm~\ref{alg:ccdmatching}, \eqref{eq:additiveupdateexact} corresponds to the width updates in lines 6 and 11.
Updating $w_n$ to $w_{n+1}$ using multiplications is equivalent to updating $s_n$ to $s_{n+1}$ using subtractions.

However, we observe that \eqref{eq:additiveupdateexact} requires the computation of $\log$ function with FP which is typically no less complex than a multiplication.
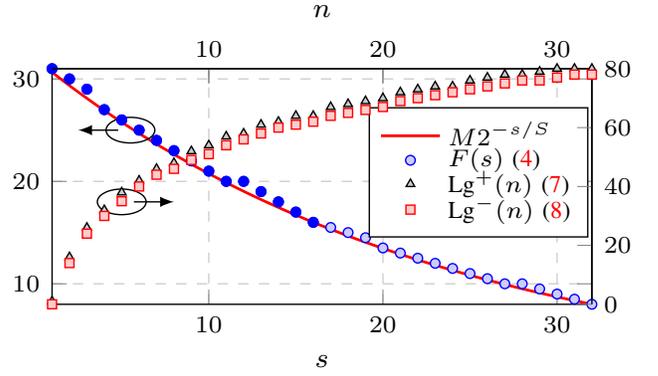
\begin{figure}[t]
\centering
\resizebox{\columnwidth}{!}{
\begin{tikzpicture}
\begin{axis}[
  every axis/.append style={font=\footnotesize},
height=0.45\columnwidth,
width=0.8\columnwidth,
xmin=1,
xmax=32,
ymin=8,
ymax=31,
grid style={dashed,lightgray!75},
xlabel={$\step$},
grid=major,
ticklabel style = {font=\scriptsize},
ylabel near ticks,
xlabel near ticks,
legend style={at={(0.79,0.56)},anchor=center,font=\scriptsize,legend cell align=left,row sep=-0.98ex},
]
\addlegendimage{color=red, thick}
\addlegendentry{$M2^{-s/S}$}
\addlegendimage{color=blue, only marks, mark=*, mark size=1.4pt,mark options={fill=blue!20!white,solid}}
\addlegendentry{$F(s)$ \eqref{eq:expfunc}}
\addlegendimage{color=black, only marks, mark=triangle*, mark size=1.45pt,mark options={fill=black!20!white,solid}}
\addlegendentry{$\lgp(n)$ \eqref{eq:lgplus}}
\addlegendimage{color=red, only marks, mark=square*, mark size=1.3pt,mark options={fill=red!20!white,solid}}
\addlegendentry{$\lgm(n)$ \eqref{eq:lgminus}}
\addplot [color=blue, only marks, mark=*, mark size=1.4pt,mark options={fill=blue!20!white,solid}]
  table[row sep=crcr]{
  1	31\\
2	30\\
3	29\\
4	27\\
5	26\\
6	25\\
7	24\\
8	23\\
9	22\\
10	21\\
11	20\\
12	20\\
13	19\\
14	18\\
15	17\\
16	16\\
17	15.5000000000000\\
18	15\\
19	14.5000000000000\\
20	13.5000000000000\\
21	13\\
22	12.5000000000000\\
23	12\\
24	11.5000000000000\\
25	11\\
26	10.5000000000000\\
27	10\\
28	10\\
29	9.50000000000000\\
30	9\\
31	8.50000000000000\\
32	8\\
};
\addplot [color=red, thick, domain=-10:40, samples=100] {32*2^(-x/16)};
\addplot [color=blue, only marks, mark=*, mark size=1.4pt,mark options={fill=blue,solid}]
  table[row sep=crcr]{
  1	31\\
2	30\\
3	29\\
4	27\\
5	26\\
6	25\\
7	24\\
8	23\\
9	22\\
10	21\\
11	20\\
12	20\\
13	19\\
14	18\\
15	17\\
16	16\\
};
\draw (5.5,25) ellipse (1.4 and 1.25) node[anchor=center] (aa) {};
\draw[draw,-latex] (aa) -- (2.5,25);
\draw (5,18) ellipse (1.4 and 1.25) node[anchor=center] (bb) {};
\draw[draw,-latex] (bb) -- (8,18);
\end{axis}
\begin{axis}[
xmin=0,
xmax=1,
every axis/.append style={font=\footnotesize},
height=0.45\columnwidth,
width=0.8\columnwidth,
xmin=1,
xmax=32,
ymin=0,
ymax=80,
ylabel near ticks,
xlabel={$n$},
axis y line*=right,
axis x line*=top,
ticklabel style = {font=\scriptsize},
legend style={legend pos=north east,font=\scriptsize,legend cell align=left,row sep=-0.5ex},
]
\addplot [color=black, only marks, mark=triangle*, mark size=1.45pt,mark options={fill=black!20!white,solid}]
  table[row sep=crcr]{
1	1\\
2	16\\
3	26\\
4	32\\
5	38\\
6	42\\
7	46\\
8	48\\
9	52\\
10	54\\
11	57\\
12	58\\
13	61\\
14	62\\
15	64\\
16	64\\
17	67\\
18	68\\
19	69\\
20	70\\
21	72\\
22	73\\
23	74\\
24	74\\
25	76\\
26	77\\
27	78\\
28	78\\
29	79\\
30	80\\
31	80\\
32	80\\
};
\addplot [color=red, only marks, mark=square*, mark size=1.3pt,mark options={fill=red!20!white,solid}]
  table[row sep=crcr]{
1	0\\
2	14\\
3	24\\
4	30\\
5	35\\
6	40\\
7	44\\
8	46\\
9	49\\
10	51\\
11	54\\
12	56\\
13	58\\
14	60\\
15	61\\
16	62\\
17	64\\
18	65\\
19	66\\
20	67\\
21	69\\
22	70\\
23	71\\
24	72\\
25	73\\
26	74\\
27	75\\
28	76\\
29	76\\
30	77\\
31	78\\
32	78\\
};
\end{axis}
\end{tikzpicture}
}
\caption{Exponential function $F(s)$ for $1 \leq s \leq 2S$, and log functions $\lgp$ and $\lgm$ for $1\leq n \leq N$ where $\stepb=16$, $M=32$, and $N=32$.
$F(s)$ can be computed for any positive integer $s$ by storing $F(s)$ only for $1\leq s\leq S$ (filled circles) in a LUT.
$\lgp$ and $\lgm$ can be stored in two LUTs.}
\label{fig:exptable}
\end{figure}
To circumvent this complexity, we define two logarithmic functions (s.t. means such that)
\begin{IEEEeqnarray}{rCl}
\lgp(n) &\define& \max_{s=1, 2,\dotsc, S} \min_{\substack{\Delta = 0, 1,\dotsc \\ F(s)\geq nF(s+\Delta)}} \{\Delta\}, \label{eq:lgplus} \\
\lgm(n) &\define& \min_{s=1, 2,\dotsc, S} \max_{\substack{\Delta = 0, 1,\dotsc \\ F(s)\leq nF(s+\Delta)}} \{\Delta\} , \label{eq:lgminus}
\end{IEEEeqnarray}
for $n = 1, 2,\dotsc, \n$.
This computation can be realized offline. 
The resulting functions can be stored in two LUTs, each having $N$ entries.
Examples of $\lgp(n)$ and $\lgm(n)$ are shown in Fig.~\ref{fig:exptable}.

We approximate multiplications and divisions involving $F(s)$ and an integer $n$ by
\begin{IEEEeqnarray}{rCl}
    nF(\step) &\approx& F\left(\step-\lgp(n)\right), \label{eq:additiveupdate1}\\
    \frac{F(\step)}{n} &\approx& F\left(\step+\lgm(n)\right), \label{eq:additiveupdate2}
\end{IEEEeqnarray}
resp.
The relations in \eqref{eq:additiveupdate1}--\eqref{eq:additiveupdate2} are approximate due to the way $\lgp$ and $\lgm$ are defined in \eqref{eq:lgplus}--\eqref{eq:lgminus}, and due to the ceiling operation in \eqref{eq:expfunc}.
Finally, combining \eqref{eq:subdivideF} and \eqref{eq:additiveupdate1}--\eqref{eq:additiveupdate2},
\begin{equation}
F(s_{n+1}) =
\begin{cases}
F\left(s_n-\lgp(n_1)+\lgm(n_0+n_1)\right),\\
F\left(s_n-\lgp(n_0)+\lgm(n_0+n_1)\right),
\end{cases}
    \label{eq:logupdate}
\end{equation}
for symbols 1 and 0, resp.
This way, the FP-CCDM in Algorithm~\ref{alg:ccdmatching} is converted into the Log-CCDM in Algorithm~\ref{alg:ccdmatchinglogtable} by:
(1) representing the subinterval width $w_{n+1}$ by $F(s_{n+1})$, and (2) keeping track of this width through $s_{n+1}$ which is updated via \eqref{eq:logupdate} in lines 6 and 11 (highlighted in \textcolor{blue}{blue}).
The rest of the algorithms are identical.

%%%%%%%%%%%%%%%%%%%%%%%%%%%%%%
\begin{algorithm}[t]
\KwIn{ Index $\uv$, composition $C = [n_0, n_1]$ }
\KwOut{ Sequence $\ux \in \{0, 1\}^\n$}
{\it Initialize:} $s_0 \gets 1$, $I_0\gets d(\uv)$, $N\gets n_0+n_1$ \\
\For{$n = 0, 1,\dotsc, \n-1$}{
\uIf{$I_n \geq  F\left(s_n - \text{Lg}^{+}(n_0) + \lgm (\n-n) \right)$}{
$x_{n+1} \gets 1$ \\
{\color{red}$I_{n+1} \gets I_n - F\left(s_n - \text{Lg}^{+}(n_0) + \lgm (\n-n) \right)$} \\
{\color{blue}$s_{n+1} \gets s_n - \text{Lg}^{+}(n_1) + \lgm (\n-n)$} \\
$n_1 \gets n_1-1$ \\
}
\Else{
$x_{n+1} \gets 0$ \\
$I_{n+1} \gets I_{n}$ \\
{\color{blue}$s_{n+1} \gets s_n - \lgp(n_0) + \lgm (\n-n)$}\\
$n_0 \gets n_0-1$ \\
}
}
\KwRet{$\ux = (x_1, x_2,\dotsc, x_{\n})$}
\caption{Log-CCDM (Matching)}
\label{alg:ccdmatchinglogtable}
\end{algorithm}
%%%%%%%%%%%%%%%%%%%%%%%%%%%%%%

%%%%%%%%%%%%%%%%%%%%%%%%%%%%%%%%%%%%
\subsection{Representability of Data Sequences}\label{ssec:represent}
%%%%%%%%%%%%%%%%%%%%%%%%%%%%%%%%%%%%
When AC is used for compression, if the arithmetic operations are implemented with finite-precision (or in any approximate way), the {\it decodability} of the code becomes questionable.
Decodability is ensured only when the subdivisions {\it do not create overlapping subintervals}, which would make two different source sequences (the red circles in Fig.~\ref{fig:OverlapEffect}) mapped to the same code sequence (the blue circle).
This makes uniquely decodability during decompression impossible~\cite{Rissanen1979AC}.
On the other hand, gaps between subintervals only decrease the efficiency of the source code, i.e., increases its rate.

When AC is used for DM which is a dual of source coding, see Sec.~\ref{sec:intro}, the dual of decodability is {\it representability}~\cite{Martin1983ACforconstrainedchan}.
AC creates an invertible mapping from data sequences to CC sequences only if the {\it subintervals do not have gaps} in between.
Overlaps are allowed.
As discussed above via Fig.~\ref{fig:OverlapEffect}, overlaps make two amplitude sequences mapped to the same data sequence. 
However, since matching is realized first in the case of DM (at the transmitter), this does not create a problem: some channel sequences are just never generated.
This only decreases the efficiency of the code, i.e., decreases its rate.

With an FP implementation, AC creates subintervals that neither overlap nor have gaps, satisfying the decodability and representability criteria as in Fig.~\ref{fig:MatchExample}.
For our Log-CCDM algorithm, the representability, i.e., ``no gaps'', criterion can be expressed from \eqref{eq:logupdate} as~\cite[eq. (14)]{Martin1983ACforconstrainedchan}
\begin{IEEEeqnarray}{rCl}
F(s_{n}) &\leq& F\left(s_n-\lgp(n_1)+\lgm(n_0+n_1)\right) \nonumber \\ && + \hspace{0.1cm} F\left(s_n-\lgp(n_0)+\lgm(n_0+n_1)\right). \label{eq:overlapcondition}
\end{IEEEeqnarray}
%%%%%%%========
\begin{lemma}
The condition in \eqref{eq:overlapcondition} is satisfied for $F(s)$, $\lgp(n)$, and $\lgm$ defined in \eqref{eq:expfunc}, \eqref{eq:lgplus}, and \eqref{eq:lgminus}, resp.
\end{lemma}
%%%%%%%========
\begin{proof}
First, observe that the following inequalities are satisfied from \eqref{eq:lgplus}--\eqref{eq:lgminus} by definition for any positive integer $s$:
\begin{IEEEeqnarray}{rCl}
F(\step-\lgp(n)) &\geq&  nF(\step), \label{eq:pluslogconsequence}\\
F(\step+\lgm(n)) &\geq& \frac{F(\step)}{n}. \label{eq:minuslogconsequence}
\end{IEEEeqnarray}
Then for the first case in \eqref{eq:logupdate}, we can write
\begin{IEEEeqnarray}{rCl}
F(s_{n+1}) &=& F\left(s_n-\lgp(n_1)+\lgm(n_0+n_1)\right) \nonumber \\ 
&\stackrel{\eqref{eq:pluslogconsequence}}{\geq}& n_1 F\left(s_n+\lgm(n_0+n_1)\right) \nonumber \\
  &\stackrel{\eqref{eq:minuslogconsequence}}{\geq}& \frac{n_1 F\left(s_n\right)}{n_0+n_1}. \label{eq:decode1}
\end{IEEEeqnarray}
Similarly, for the second case in \eqref{eq:logupdate}, we can write
\begin{equation}
 F(s_{n+1}) \geq \frac{n_0F\left(s_n\right)}{n_0+n_1}. \label{eq:decode2}
\end{equation}
Combined, \eqref{eq:decode1} and \eqref{eq:decode2} imply that \eqref{eq:overlapcondition} is satisfied.
\end{proof}
%%%%%%%========

\section{Properties of Log-CCDM}\label{sec:logccdmprop}
%%%%%%%%%%%%%%%%%%%%%%%%%%%%%%%%%%%%
\subsection{Input Length of the Matcher}
%%%%%%%%%%%%%%%%%%%%%%%%%%%%%%%%%%%%
The final width for any input in Algorithm~\ref{alg:ccdmatchinglogtable} is $F(s_N)$ where
\begin{equation}
    s_N = s_0 + \underbrace{\sum_{i=1}^{n_0+n_1} \lgm(i) - \sum_{j=1}^{n_0} \lgp(j) - \sum_{t=1}^{n_1} \lgp(t)}_{\define\gamma}.
\end{equation}
Thus, all CC sequences are represented with an interval of identical width.
This allows us to represent each interval with a fixed length $k$-bit index as discussed in Sec.~\ref{sec:ACDM}.

There are $F(s_0)/F(s_N)$ CC sequences, i.e., $|\mathcal{C}_{\text{cc}}| = F(s_0)/F(s_N)$.
Then as in Sec.~\ref{sec:ACDM},
\begin{IEEEeqnarray}{rCl}
k &=& \lf \log |\mathcal{C}_{\text{cc}}| \rf = \lf \log \frac{F(s_0)}{F(s_N)} \rf = \lf \log \frac{F(s_0)}{F(s_0+\gamma)} \rf \nonumber\\
  &\stackrel{\eqref{eq:expfunc}}{=}& \lf \log \frac{F(s_0)}{2^{-\lf\gamma/S\rf}F(s_0)} \rf = \lf \frac{\gamma}{S} \rf.
\end{IEEEeqnarray}
The parameters $\gamma$ and $k$ can be computed offline.

%%%%%%%%%%%%%%%%%%%%%%%%%%%%%%%%%%%%
\subsection{Storage Complexity and Arithmetic Precision}\label{ssec:storagearithmetic}
%%%%%%%%%%%%%%%%%%%%%%%%%%%%%%%%%%%%
The storage complexity of Log-CCDM is $S\log M + 2N(\log S + \log\log N)$~bits.
To store $F(s)$ in \eqref{eq:expfunc}, we need a LUT with $S$ entries, each stored with $\log M$ bits as discussed in Sec.~\ref{ssec:logccdm}.
Thus, the storage requirement of this LUT is $S\log M$~bits.
From \eqref{eq:expfunc} and \eqref{eq:additiveupdate1}--\eqref{eq:additiveupdate2}, it can be shown that both $\lgp(n)$ and $\lgm(n)$ are approximately equal to $S\log n$.
Therefore, the entries in the corresponding LUTs can be stored with approximately $\log S + \log\log N$ bits assuming $S$ and $N$ are integer powers of two.
Thus, the total storage requirement of these two LUTs is $2N(\log S + \log\log N)$ bits.

In line 5 of Algorithm~\ref{alg:ccdmatchinglogtable} (highlighted in \textcolor{red}{red}), values of $F(\cdot)$ are recursively subtracted from the input index.
These subtractions require an arithmetic precision of $\log M + \log N $ bits assuming again $N$ is an integer power of two. 
The proof of this will consist of a lemma and a theorem.
%%%%%%%========
\begin{lemma}\label{lem2}
If $I_n < F(s_n)$, then Algorithm~\ref{alg:ccdmatchinglogtable} guarantees that $I_{n+1} < F(s_{n+1})$.
This implies that if $I_0 < F(s_0)$, then all $I_n$ for $n=0, 1,\dotsc, N$ satisfy $I_n<F(s_n)$.
\end{lemma}
%%%%%%%========
\begin{proof}
We define $s_{n,0} \define s_n - \lgp(n_0) + \lgm(N-n)$ and $s_{n,1} \define s_n - \lgp(n_1) + \lgm(N-n)$ for $n = 0, 1,\dotsc, N-1$.
Note that $s_{n,0}$ and $s_{n,1}$ are the candidates for $s_{n+1}$ in lines 11 and 6 of Algorithm~\ref{alg:ccdmatchinglogtable}.
We first observe that
\begin{equation}
    I_n < F(s_n) \stackrel{\eqref{eq:overlapcondition}}{\leq} F(s_{n,0}) + F(s_{n,1}). \label{eq:lemma1}
\end{equation}
Then there are two options at the $n^{\text{th}}$ step of Algorithm~\ref{alg:ccdmatchinglogtable}:
\begin{itemize}
    \item 
    If $I_n \geq F(s_{n,0})$ (line 3), then $s_{n+1}=s_{n,1}$ and
    \begin{IEEEeqnarray}{rCl}
    I_{n+1} &=& I_n - F(s_{n,0}) \nonumber \\
    &\stackrel{\eqref{eq:lemma1}}{<}& F(s_{n,0}) + F(s_{n,1}) - F(s_{n,0}) = F(s_{n+1}). \nonumber
    \end{IEEEeqnarray}
    
    \item 
    If $I_n < F(s_{n,0})$ (line 8), then $s_{n+1}=s_{n,0}$ and
     \begin{IEEEeqnarray}{rCl}
     I_{n+1} = I_n < F(s_{n,0}) = F(s_{n+1}). \nonumber 
     \end{IEEEeqnarray}
\end{itemize}
Since $I_0 < 1 \leq F(s_0=1)$ due to \eqref{eq:expfunc} and Algorithm~\ref{alg:ccdmatchinglogtable}, we conclude that $I_n < F(s_n)$ for $n = 0, 1,\dotsc, N$.
We note that this proof is similar to the proof of~\cite[Lemma 1]{GultekinWHS2018_ApproxEnumerative}.
\end{proof}
%%%%%%%========
\begin{theorem}
The subtraction $I_n-F(s_{n,0})$ in line 5 of Algorithm~\ref{alg:ccdmatchinglogtable} requires an arithmetic precision of at most $\log M + \log N$ bits assuming $N$ is an integer power of two.
\end{theorem}
%%%%%%%========
\begin{proof}
First, observe from \eqref{eq:decode1} that
\begin{equation}
   \frac{F(s_{n,0})}{F(s_n)} \geq \frac{n_0}{n_0+n_1} \geq \frac{1}{\n}.
\end{equation}
Thus, since $I_n < F(s_n)$ from Lemma~\ref{lem2}, the subtrahend of the subtraction $I_n - F(s_{n,0})$ in line 5 of Algorithm~\ref{alg:ccdmatchinglogtable} can be at most $N$ times smaller than its minuend.
Since each entry of the LUT of $F(\cdot)$ is $\log M$-bit long, this subtraction requires an arithmetic precision of at most $\log M + \log N$~bits.
\end{proof}
%%%%%%%========

%========
\begin{remark}
Log-CCDM can be implemented with two $(\log M + \log N)$-bit shift registers.
The $k$-bit input index (and some trailing zeros if necessary) is gradually loaded into the first register, which stores the minuend.
The $\log M$-bit $F(\cdot)$ values are loaded into the second register, which stores the subtrahend.
The dematching can be implemented with the same principle via $(\log M + \log N )$-bit additions.
\end{remark}
%========
%========
\begin{remark}
It is not straightforward to evaluate the complexity of FP-CCDM~\cite[Sec. 6]{Gultekin2019Arxiv_Comparison}.
In this work, we have transformed FP-CCDM, which requires no storage but is based on multiplications and divisions, into Log-CCDM, which requires a small amount of storage but is based on additions and subtractions.
Then the selection among these two depends on how costly the multiplications and divisions are in comparison to a certain amount of storage for given hardware.
\end{remark}
%========

%%%%%%%%%%%%%%%%%%%%%%%%%%%%%%%%%%%%
\subsection{Rate Loss}\label{ssec:rateloss}
%%%%%%%%%%%%%%%%%%%%%%%%%%%%%%%%%%%%
There are two sources of inaccuracies in Log-CCDM.
The first is due to \eqref{eq:expfunc} which leads to an imprecise representation of the log of the interval width.
The second is due to \eqref{eq:lgplus} and \eqref{eq:lgminus}, and then to \eqref{eq:pluslogconsequence} and \eqref{eq:minuslogconsequence}  which lead to imprecise multiplications (with $n_0$ or $n_1$) and divisions (with $n_0+n_1$).
These two inaccuracies lead to overlapping intervals as discussed in Sec.~\ref{ssec:represent}.
Thus, some CC sequences that would be generated by FP-CCDM are never produced by Log-CCDM, which decreases the rate $k/\n$, causing a rate loss.

In Fig.~\ref{fig:ratelosses}, we show $k/\n$ as a function of $\n$ for a composition of $C=[0.75\n, 0.25\n]$ for FP-CCDM and Log-CCDM with different $(S, M)$ pairs.
First, we see that FP-CCDM is asymptotically optimum for large $\n$ as shown in~\cite{SchulteB2016_CCDM}.
This is in the sense that it has a matching rate $\define k_{\max}/\n$ converging to the (binary) entropy $H(0.75)=0.8113$ bits of the resulting distribution, see the black curve in Fig.~\ref{fig:ratelosses}.
However, Log-CCDM is suboptimal, i.e., it does not converge to the entropy.
By increasing $S$ and $M$, this rate can be made closer to the entropy, while the required storage increases, see Sec.~\ref{ssec:storagearithmetic}.

Next, we see that for increasing $S$ and $M$, the matching rate of Log-CCDM gets closer to $k_{\max}/\n$.
When using LUTs with $S=512$ and $N\leq 1024$ entries (the red circles), the gap to $k_{\max}/\n$ is around 0.007 bit/sym.
The gap increases to 0.03 bit/sym for LUTs with $S=128$ and $N\leq 1024$ entries (the blue curve with squares).
Thus, there is a trade-off between the rate loss with respect to $k_{\max}/N$ and the table sizes.
Assume we operate at $N=1024$, $S=512$, $M=1024$, and $k/\n=0.7988<0.8063=k_{\max}/\n$ bit/sym (the filled red circle).
Then the required storage for the three LUTs is $S\log M + 2N(\log S + \log\log N) = 3.79$~ kBs (bottom-right inset figure), and the required arithmetic precision is $\log M + \log N=20$ bits (top-left inset figure), see Sec.~\ref{ssec:storagearithmetic}.

\begin{figure}[t]
\centering
	\resizebox{0.9\columnwidth}{!}{\includegraphics{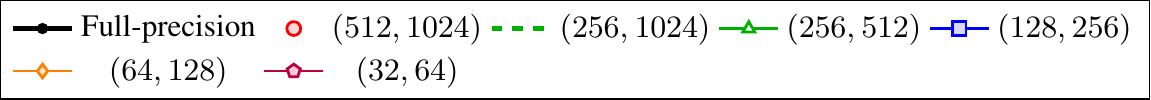}}
	\resizebox{0.9\columnwidth}{!}{\includegraphics{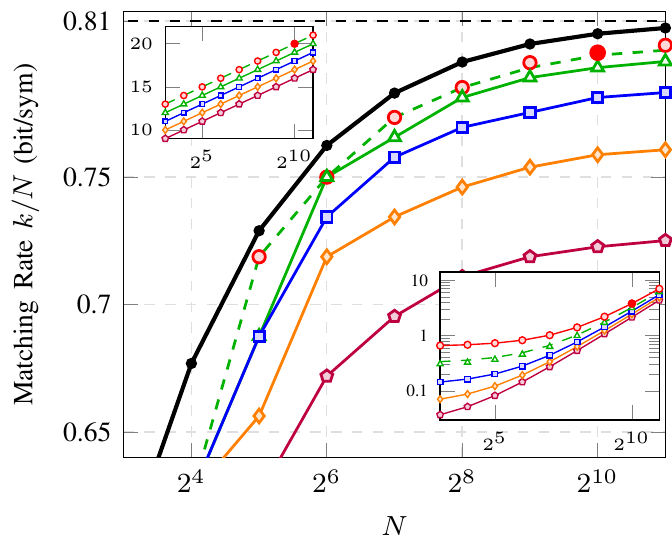}}
\caption{Rate $k/N$ of Log-CCDM with $(S, M)$. {\bf Top-left inset:} Required arithmetic precision in bits. {\bf Bottom-right inset:} Required storage in kBs.}
\label{fig:ratelosses}
\end{figure}

Last, we see that for $S=256$, the required storage is virtually identical with $M=512$ and $1024$ (bottom-right inset figure). 
However, the rate loss of $(S=256, M=1024)$ (dashed green) is smaller.
For instance, the same rate loss as $(S=512, M=1024, N=1024)$ (filled red circle) is obtained via $(S=256, M=1024, N=1024)$ (dashed green), requiring $3.22$ kBs of storage instead of $3.79$.
This demonstrates that $M$ plays a more important role in determining the rate loss.
For a given $M$, a smaller $S$ can be chosen to decrease the storage complexity, since the rate loss is relatively insensitive to the changes in $S$.

%%%%%%%%%%%%%%%%%%%%%%%%%%%%%%%%%%%%%%%%%%%%%%%%%%%%%%
\section{Conclusion}\label{sec:conclusions}
%%%%%%%%%%%%%%%%%%%%%%%%%%%%%%%%%%%%%%%%%%%%%%%%%%%%%%
In this paper, we have introduced Log-CCDM: an approximate algorithm to realize arithmetic-coding-based constant composition distribution matching. Log-CCDM operates in the logarithmic domain and is based on three simple lookup tables (LUTs). 
Thanks to these LUTs, it is possible to realize CCDM (1) with an arithmetic precision that grows logarithmically with input length (instead of linearly), and (2) only using additions and subtractions (instead of multiplications and divisions).
This decreases the computational complexity, however, increases the storage complexity (to a few kilobytes) due to the LUTs.
The performance of Log-CCDM in terms of rate depends on the sizes of these LUTs.

\clearpage
{ {\bf Acknowledgements:} The work of Yunus Can G\"{u}ltekin and Alex Alvarado has received funding from the European Research Council under the European Union's Horizon 2020 research and innovation programme via the Starting grant FUN-NOTCH (grant agreement ID: 757791) and via the Proof of Concept grant SHY-FEC (grant agreement ID: 963945).}

\bibliographystyle{IEEEtran}
\bibliography{IEEEabrv,PhD_refs}

\end{document}